\documentclass[journal]{IEEEtran}

\usepackage[dvips]{color}
\usepackage{epsf}
\usepackage{times}
\usepackage{epsfig}
\usepackage{graphicx}
\usepackage{amsmath}
\usepackage{amssymb}
\usepackage{amsxtra}
\usepackage{here}
\usepackage{rawfonts}
\usepackage{times}
\usepackage{url}
\usepackage{cite}
\newtheorem{theorem}{\bf Theorem}

\newtheorem{definition}{\bf Definition}
\newtheorem{remark}{\bf Remark}
\IEEEoverridecommandlockouts
\begin{document}
\title{\huge Load Shifting in the Smart Grid: To Participate or Not?\vspace{-0.2cm}}

\author{\authorblockN{Yunpeng Wang$^{1,2}$, Walid Saad$^{3,4}$, Narayan B. Mandayam$^5$, H. Vincent Poor$^6$}\\ \authorblockA{\small%$^\textbf{2}$ 
$^1$ Electrical and Computer Engineering Department, University of Miami, Coral Gables, FL\\
$^2$ College of Computer and Software, Nanjing University of Information Science and Technology, China\\
$^3$ Wireless@VT, Bradley Department of Electrical and Computer Engineering,  Virginia Tech, Blacksburg, VA\\
$^4$ Department of Computer Engineering, Kyung Hee University, South Korea\\
$^5$ WINLAB, Dept. of ECE, Rutgers University, North Brunswick, NJ\\
$^6$ Electrical Engineering Department, Princeton University, Princeton, NJ\\
Email: \url{y.wang68@umiami.edu}, \url{walids@vt.edu},  \url{narayan@winlab.rutgers.edu}, \url{poor@princeton.edu}
\vspace{-0.9cm}
}%

\thanks{This research was supported by the U.S. National Science Foundation under Grants CNS-1446621, ECCS-1549894, ECCS-1549900, ECCS-1549881, and NSFC-61232016. Dr. Saad was a corresponding author. }
 }
\date{}
\maketitle

\begin{abstract}
Demand-side management~(DSM) has emerged as an important smart grid feature that allows utility companies to maintain desirable grid loads. However, the success of DSM is contingent on active customer participation. Indeed, most existing DSM studies are based on game-theoretic models that assume customers will act rationally and will voluntarily participate in DSM. In contrast, in this paper, the impact of customers' subjective behavior on each other's DSM decisions is explicitly accounted for. In particular, a noncooperative game is formulated between grid customers in which each customer can decide on whether to participate in DSM or not. In this game, customers seek to minimize a cost function that reflects their total payment for electricity. Unlike classical game-theoretic DSM studies which assume that customers are rational in their decision-making, a novel approach is proposed, based on the framework of prospect theory (PT), to explicitly incorporate the impact of customer behavior on DSM decisions. To solve the proposed game under both conventional game theory and PT, a new algorithm based on fictitious player is proposed using which the game will reach an $\epsilon$-mixed Nash equilibrium. Simulation results assess the impact of customer behavior on demand-side management. In particular, the overall participation level and grid load can depend significantly on the rationality level of the players and their risk aversion tendency.

\vspace{-0cm}
\end{abstract}
\begin{keywords}
Smart grid, game theory, prospect theory, demand side management.
\end{keywords}\vspace{-0.1cm}

\section{Introduction}\vspace{-0cm}
The use of demand-side management~(DSM) mechanisms has recently attracted significant attention in the smart grid literature. DSM schemes offer smart grid customers an opportunity to change their demands over time so as to reduce their overall electricity costs. From the utility company's perspective, such a shifting of demand can reduce the peak hour demand on the grid~\cite{hossain2012smart}. A successful implementation of DSM schemes requires customers to actively subscribe to the offers made by the utility company. However, recent empirical studies have shown that customers remain reluctant to subscribe to DSM mechanisms, despite the efforts of utility companies~\cite{FAHEY}. Therefore, it is important to develop a new generation of DSM mechanisms that can improve the penetration of the technology and thus accelerate the deployment of the smart grid.
 
There has been an abundant body of work dealing with DSM~\cite{mohsenian2010autonomous, atzeni2013noncooperative, QZ00, fadlullah2014gtes, chen2014autonomous, chai2014demand}. The authors in~\cite{mohsenian2010autonomous} proposed a distributed, game-theoretic DSM system, in which each user chooses a daily schedule of household appliances and loads to optimize the energy consumption. The work in~\cite{atzeni2013noncooperative} proposed different game-theoretic approaches to optimize a day-ahead DSM mechanism by using storage units. The overall value of implementing DSM and demand response schemes was studied in~\cite{QZ00} via a Stackelberg game formulation. The work in~\cite{fadlullah2014gtes} investigated how energy consumption may be optimized through a two-step centralized model, in which a power supplier provided consumers with an energy price parameter and consumption summary vector. In~\cite{chen2014autonomous}, the authors adopted an instantaneous load billing scheme so as to shift consumers' peak hour demand and to charge them fairly for their energy consumption. In~\cite{chai2014demand}, utility companies and residential users were modeled in two levels, which reduce demand variation and peak load.

Thus, the works in~\cite{mohsenian2010autonomous, atzeni2013noncooperative, QZ00, fadlullah2014gtes, chen2014autonomous, chai2014demand} study the various economic and optimization aspects of DSM and are representative of the majority of existing works in this area. However, most of those existing works assume that customers will act rationally and subscribe to DSM as long as their objective function can be optimized~\cite{mohsenian2010autonomous, QZ00, fadlullah2014gtes, chai2014demand, atzeni2013noncooperative, chen2014autonomous}. In practice, as observed by numerous real-world experimental studies~\cite{fiegenbaum1988attitudes, barberis2001prospect, levy1997prospect}, the process of individuals can deviate significantly from the rational premise of conventional game theory. This has been corroborated by the relatively low adoption rate of DSM in deployed smart grids~\cite{FAHEY}. In this respect, prospect theory (PT), a Nobel-prize winning theory developed by Kahneman and Tversky, provides the necessary tools to better understand real-world decision making and its deviation from rational behavior~\cite{kahneman1979prospect}. In particular, PT studies have shown that, in real life, users often have subjective behavior when faced with uncertainty of outcome, such as lottery outcomes. Also, it has been shown that customers may have subjective perceptions on how others behave and how this behavior impacts their gains or losses in economic-oriented scenarios. Since the price in DSM strongly depends on aggregate demands and thus on the participation of all customers, the overall perceptions of the customers on each other's DSM decisions will impact their behavior. PT has been widely used in the social sciences~\cite{harrison2009expected}. In addition some recent works~\cite{6310922, Mandayam, okuda2009design, poortoappear} have applied PT to wireless networks. However, to the best of our knowledge, beyond our early work in~\cite{PTICASSP} that studies a two-player storage management PT game, no existing work seems to have investigated how PT considerations impact participation in DSM.

The main contribution of this paper is to propose a new framework for DSM that accounts for realistic customer participation behavior using prospect theory. In particular, the DSM problem is cast as a noncooperative game between customers. In this game, each customer can decide whether or not to participate in a DSM program, while the utility company seeks to reduce the total peak hour demand so as to maintain a desirable target load on the grid. The proposed game uses PT to explicitly incorporate the customers' \emph{subjective perceptions} on DSM decisions and their impact on potential cost savings. Here, subjective perceptions pertain to the way in which each customer evaluates its electricity payment and how that depends on other customers' actions. In this respect, each customer is seeking to minimize a cost function that reflects its one-day electricity bill, under other customers' participation and its impact on the price.  

Compared to related works on DSM~\cite{mohsenian2010autonomous, atzeni2013noncooperative, QZ00, fadlullah2014gtes, chen2014autonomous, chai2014demand, saad2012game, 6787063, weaver2009game, coogan2013energy, cui2013game}, this paper brings forward novel ideas from PT in order to explicitly account for realistic customer behavior, which can differ significantly from the classical, rational path predicted by traditional game theory. To solve the game, under both PT and classical game theory, we propose a new learning algorithm that allows the customers to interact and reach an $\epsilon$-mixed Nash equilibrium. We prove the convergence of the algorithm and discuss the properties of the reached operating point. Extensive simulation results based on realistic data show that deviations from rational behavior can strongly impact the overall level of participation in DSM.

The remainder of the paper is organized as follows: Section~\ref{sec:sysmodel} presents the studied system model. In Section~\ref{sec:game}, we formulate the problem as a PT-based game. In Section~\ref{sec:algo}, we introduce the concept of fictitious play and describe our proposed algorithm. Simulation results are presented in Section~\ref{sec:sim}, while conclusions are drawn in Section~\ref{sec:conc}.\vspace{-0cm}

\section{System Model}\label{sec:sysmodel}
\subsection{Demand cost}\label{sec:dp}
Consider a smart grid consisting of a set $\mathcal{N}$ of customers in which each customer $i \in \mathcal{N}$ consumes a certain amount of energy per hour. All customers are offered the opportunity to participate in a demand side management scheme provided by the utility company. For customer $i$, we define an hourly \emph{energy demand scheduling vector} in line with existing works such as,
\begin{equation}\label{eq:x}
\boldsymbol{x}_i=[x_i^1, x_i^2,\ldots, x_i^H],
\end{equation}
where $H=24$ hours. For a certain hour $h \in \mathcal{H}=\{1,2,\ldots, H\}$, $x_i^h$ represents the energy demand of customer $i$. The total energy demand from all customers at time $h$ is thus
\begin{equation}\vspace{-0.1cm}\label{eq:actuald}
d^h=\sum_{i \in \mathcal{N}} x_i^h.
\end{equation}

At a given time $h$, we assume that the price per energy unit charged to a customer $i$ is dependent on its fraction of the total current load as follows~\cite{mohsenian2010optimal}:
\begin{equation}\label{eq:price}
c_i^h(\boldsymbol{x})=B \frac{x_i^h}{\sum_{i \in \mathcal{N}} x_i^h},
\end{equation}
where $B$ could be designed based on notions from a locational marginal pricing (LMP) scheme~\cite{powerbookLMP} in which the price function is not necessarily time-dependent. The electricity price for each customer as given by (\ref{eq:price}) allows one to allocate the price for the amount of energy that a customer consumes. For example, as the number of customers increases, the total demand will increase. Then, the electricity pricing must reflect this change in the demand. For each customer, the individual electricity pricing is dependent on its demand proportion as captured by (\ref{eq:price}). Hence, the total cost of user $i$ over $H$ hours is given by
%LMP is a mechanism for using market-based prices for managing transmission congestion. A power company will announce such LMP based on time, level and generation. As the electricity for each customer in (\ref{eq:price}), it will be allocated and distributed based on the amount of energy that a customer consumes. Thus, individual prices are sent to each customer, and not a fixed price for all. 
\begin{equation}
\sum_{h=1}^H c_i^h(\boldsymbol{x}) x_i^h.
\end{equation}

Given the price and local demand, an energy market is set up in which all customers seek to minimize their costs while maintaining their energy demands at a desired level. Here, all customers will interact so as to determine their demands under DSM. These demands include the quantities of energy required and impact the price at a certain time. Instead of fixed reservation prices announced by a utility company, in our model, each user can strategically change its demand, in which its demand fraction of the total demand impacts the underlying electricity price. In this respect, the price in (\ref{eq:price}) can lead to an increased cost for all customers, including those that decide not to participate in load management. Moreover, the demand delayed by DSM will cause a varying electricity price in subsequent hours. Thus, we will develop a DSM mechanism based on load shifting.
\vspace{-0.1cm}
\subsection{Load Management}\label{sec:ls}

In this subsection, we propose a load shifting mechanism to analyze the variations of the demand over time. Inherently, load shifting allows part of the peak hour load to be moved to an off-peak hour, in which such controlled demand is delayed to the following time slot in our formulation. For the proposed model, we assume that, for each participating customer, the demand can be adjusted so as to meet a \emph{predefined/target energy demand} profile (vector) set by the utility company, given by
\begin{equation}\label{eq:gd}\vspace{-0.1cm}
\boldsymbol{G_d}=[g_d^1, g_d^2,\ldots, g_d^H].
\end{equation}
On the one hand, a utility company wants to reduce peak hour demand in order to decrease load on the grid. On the other hand, the company 
wants to maintain the amount of load shifted within a reasonable range while avoiding the creation of a new peak hour. For example, if the power company wants to reduce $10\%$ of the load at a given hour $\hat{H}$, the target energy demand could be defined by

\begin{equation}\label{eq:gd2}\vspace{-0.1cm}
g_d^h=\begin{cases}
g^h, &\text{if } h \neq \hat{H},\\
\beta g^h, &\text{if } h = \hat{H},
\end{cases}
\end{equation}
where $\beta=0.9$ and $g_h \in [g^1, g^2,\ldots, g^H]$ can be the \emph{historical demand} referenced by the utility company. Here, we assume that a customer can determine when it starts to participate in DSM. In this respect, customer $i$ will choose its starting time (defined as the \emph{participating time} $a_i$) so as to minimize its cost by observing the difference between the predefined demand per user and its daily demand. In particular, we assume that a customer $i$ would not leave DSM after choosing to participate over $H=24$ hours and thus, the number of participating customers at a given time $h$ is $I^h \triangleq |\mathcal{I}^h|=|\{a_i \ge h\}|$. Here, $a_i$ is the starting hour and we discuss its impact in more detail in Section~\ref{sec:gameconcept}. In this case, the reduced demand of participating customer $i$ is given by
\begin{equation}\label{eq:redde}\vspace{-0.1cm}
r_i^h=\begin{cases}
\gamma_i(x_i^h-\frac{g_d^h-\sum_{i \in \mathcal{N}\setminus \mathcal{I}} x_i^h}{I^h})^+, &\text{if } g_d^h<l^h,\\
0, &\text{if } g_d^h\ge l^h,
\end{cases}
\end{equation}
where $(q)^+:= \max (0, q)$ and $l^h$ is the total demand that includes both $d^h$ and the amount shifted from the previous hour $h-1$, such that $l^h=d^h+\sum_{i \in \mathcal{N}} r_i^{h-1}$. $0<\gamma_i \le 1$ is a factor by which customer $i$ wants to reduce from exceeding its demand, as $\frac{g_d^h-\sum_{i \in \mathcal{N}\setminus \mathcal{I}} x_i^h}{I^h}$ is the averaged demand suggested by the utility company for all participating customers. In particular, if the demand of participating customer $i$ is less than the average demand for participating customers, i.e., $x_i^h<\frac{g_d^h-\sum_{i \in \mathcal{N}\setminus \mathcal{I}} x_i^h}{I^h}$, its reduced demand is $r_i^h=\gamma_i \cdot 0=0$.

Using (\ref{eq:redde}), if customer $i$ participates in DSM at time $h$, its demand will be 
\begin{equation}\label{eq:now}\vspace{-0.1cm}
y_i^h=x_i^h-r_i^h.
\end{equation}
Then, the participating customer $i$ moves its shifted load to the following hour, and thus, its DSM demand at time $h<t<H$ is 
\begin{equation}\label{eq:later1}\vspace{-0.1cm}
y_i^{t}=(x_i^{t}+r_i^{t-1})-r_i^{t},
\end{equation}
while the demand at the final hour $H$ is
\begin{equation}\label{eq:later2}
y_i^H=x_i^{H}+r_i^{H-1}.
\end{equation}

To analyze such a load management scheme, we next propose a new framework that builds on the analytical tools of classical game theory and prospect theory.

\vspace{-0cm}

\section{Game-theoretic Formulation for Demand-Side management}\label{sec:game}\vspace{-0cm}

In this section, we first formulate a noncooperative game between the customers, and then study the proposed load shifting model using expected utility theory~\cite{GT00} and prospect theory~\cite{kahneman1979prospect}, also discussing their various properties.\vspace{-0.1cm}

\subsection{Noncooperative Game Model}\label{sec:gameconcept}\vspace{-0cm}
In order to analyze the interactions between customers, we use noncooperative game theory~\cite{GT00}, as the strategy choices of the customers are interdependent. We formulate a strategic noncooperative DSM game $\Xi=(\mathcal{N},\{\mathcal{A}_i\}_{i\in\mathcal{N}},\{u_i\}_{i\in \mathcal{N}})$, where $\mathcal{N}$ is the set of players, the action $a_i \in \mathcal{A}_i := \{1, 2, \ldots, H\}$ of customer $i$ is the time (hour) at which customer $i$ would like to begin participation in DSM, and the cost function $u_i$ of customer $i$ captures its electricity payment to the company, using the price in (\ref{eq:price}). Here, we note that, although a customer can participate in load management, its total daily demand remains the same (i.e., $\sum_{h=1}^{H} x_i^t=\sum_{h=1}^{H} y_i^t$). Thus, the utility function (cost) for a player $i \in \mathcal{N}$ that chooses an action $a_i$ is given by
\begin{equation}\label{eq:utility}
\begin{split}
u_i(a_i, \boldsymbol{a}_{-i})=\sum_{h=1}^H c_i^h\biggl(\sum_{i \in \mathcal{N}} y_i^h(a_i,\boldsymbol{a}_{-i}) \biggr)\times y_i^h(a_i,\boldsymbol{a}_{-i}),
\end{split}
\end{equation}
where $\boldsymbol{a}_{-i}=[a_1, a_2, \dots, a_{i-1}, a_{i+1}, \dots, a_N]^T$ is the vector of action choices of all players other than $i$, and $y_i^h(a_i,\boldsymbol{a}_{-i})$ is the DSM demand of user $i$, compared to the initial demand $x_i^h$ in (\ref{eq:x}). For example, each customer can determine a starting hour $a_i$ and shift its load to after that hour.

The goal of each customer $i$ is to choose a strategy $a_i \in \mathcal{A}_i$ so as to minimize its cost as given in (\ref{eq:utility}). For characterizing a desirable outcome for the studied game $\Xi$, one must derive a common solution that enables one to capture the coupling between the individual optimization of each customer. We define $\boldsymbol{a}=(a_i, \boldsymbol{a}_{-i})$ as the vector of all players' strategies. For each such vector $\boldsymbol{a}$, we will have a different electricity price $c_i^h(\cdot)$ in (\ref{eq:price}). Prior to finding a solution for the proposed DSM game, we will introduce expected utility theory and prospect theory.

\vspace{-0.1cm}
\subsection{Expected Utility Theory (EUT)}\vspace{-0cm}
In a smart grid, as the customers may, over a long time period, change their DSM preferences, we are interested in studying the frequency with which they choose a certain time to begin DSM participation. Therefore, we mainly study the proposed game under \emph{mixed strategies}~\cite{GT00} so as to capture the customers long-term, probabilistic choices of a DSM start time. Let $\boldsymbol{p}= [\boldsymbol{p}_1,\ldots,\boldsymbol{p}_N]$ be the vector of all mixed strategies, where, for every customer $i\in \mathcal{N}$, we have $\boldsymbol{p}_i=[p_i(1),\ldots, p_i(H)]^T$ and $p_i(a_i)$ is the probability corresponding to the choice of a pure strategy $a_i \in \mathcal{A}_i$. 

Under the conventional EUT model, the expected cost of customer $i$ is captured via the expected value over its mixed strategy. Computing each user's utility requires the vector of all players' strategies. In particular, we assume that the smart grid communication infrastructure will make this information available to users who participate in DSM. Thus, the EUT cost of a player $i$ will be given by \vspace{-0cm}
\begin{equation} \label{eq:multiplayerET}\vspace{-0cm}
\begin{split}
&U_i^{\text{EUT}}( \boldsymbol{p})=\sum_{\boldsymbol{a} \in \mathcal{A}}\bigg(\prod_{l=1}^N p_l(a_l)\bigg) u_i(a_i, \boldsymbol{a}_{-i}),
\end{split}
\end{equation}
where $\boldsymbol{a}$ is the vector of all players' strategies and $\mathcal{A}=\mathcal{A}_1 \times \mathcal{A}_2 \times \dots \times \mathcal{A}_N$. In particular, the mixed-strategy $\boldsymbol{p}(h)=[p_i(h), \boldsymbol{p}_{-i}(h)]^T$ represents the empirical frequencies with which the customers choose a certain participation time $h$. These frequencies capture how often customer $i$ participates in DSM. $\boldsymbol{p}_{-i}(h)$ is the vector of mixed strategies of all players other than $i$, corresponding to $\boldsymbol{a}_{-i}$ in (\ref{eq:utility}).

\vspace{-0.1cm}
\subsection{Prospect Theory}\label{sec:pt}\vspace{-0cm}

As previously mentioned, a player can evaluate its expected utility by using (\ref{eq:multiplayerET}), in which case the customers are assumed to act rationally and objectively under EUT. However, in real life, individuals do not truly behave rationally nor do they trust the rationality of others' behavior. Thus, in order to develop a realistic model of DSM, one must account for the fact that, in practice, customers may not assess their utilities objectively. Indeed, it has been shown that, despite the benefits of DSM, its adoption in practice has remained slow, due to unexpected customer behavior~\cite{FAHEY}.

To study such realistic/practical participation models, we will develop a DSM game model, in which a customer may have subjective views on how the opposing players will choose their strategies. This, in turn, impacts the way in which this customer evaluates its utility in (\ref{eq:utility}) which depends on other customers' strategies. Indeed, it has been observed that in real-life decision-making, people tend to subjectively weight uncertain outcomes~\cite{prelec1998probability}. Uncertainty, here, is defined as the fact that the amount of utility that a customer will obtain, depends on the decision making behavior of other customers, which is probabilistic. For example, a customer cannot be sure of the time at which other customers will begin their participation in DSM. Thus, the customer will not be sure how much economic gain its participation in DSM will yield; since such a participation depends on all the players' strategies. Thus, this customer's EUT evaluation as per (\ref{eq:multiplayerET}) might be \emph{overweighted} or \emph{underweighted} due to the subjective observation of others. Moreover, this customer may need to properly assess whether to shift its demand or not, as it is unsure of other customers' strategies. In this respect, a customer might have its own subjective perception on the participation of other customers in the DSM game (and on the actual price), thus deviating from the rational assumption of conventional game theory and EUT. 

In order to analyze such subjective perceptions, we will use the behavioral framework of prospect theory~\cite{kahneman1979prospect}. In this studied model, we mainly focus on how a customer evaluates the strategies of its opponents and, thus, acts accordingly. Our focus is on the prospect-theoretic perspective, which primarily deals with how human decision makers deal with economic outcomes that have some form of uncertainty.

One important PT concept is the so-called \emph{weighting effect}. Weighting effect naturally implies that customers will have subjective views on how their opponents will act (i.e., a weighted view on the probability vector of those players), which, in turn, maps to subjective views on their expected utilities. In the proposed game, we use this weighting effect to incorporate a subjective evaluation for each user's observation of the mixed strategy of its opponents. Thus, under PT, instead of objectively observing the mixed strategy vector $\boldsymbol{p_{-i}}$ chosen by the other players, each user perceives a weighted version of it, $w_i(\boldsymbol{p_{-i}})$. $w_i(\cdot)$ is a nonlinear transformation that maps an objective probability to a subjective one. PT studies have shown that most people will often overweight low probability outcomes and underweight high probability outcomes \cite{kahneman1979prospect}. Without loss of generality, we assume that all players use a similar weighting approach, such that $w_i(\cdot)=w(\cdot),\ \forall i \in \mathcal{N}$. Although many weighting functions exist, we choose the widely used Prelec function (for a given probability $\sigma$)~\cite{prelec1998probability},
\vspace{-0cm}
\begin{equation} \label{eq:weight}\vspace{-0.1cm}
w(\sigma)=\exp(-(-\ln \sigma)^\alpha),\ 0<\alpha \le 1,
\end{equation}
where $\alpha$ is a parameter used to characterize the distortion between subjective and objective probability. Note that when $\alpha=1$, the weighting effect will coincide with the conventional EUT probability. Fig.~\ref{fig:probaweight} illustrates the probability weighting effect. In this figure, we can see that the objective probability and subjective estimation insect at $p=1/e$ and the curve approximates EUT as $\alpha$ increases. Also, weighted summations are not necessarily equal to $1$ (i.e., $w(0.4)+w(0.6)<1$).

\begin{figure}[!t]
\begin{center}
 \vspace{-0.2cm}
 \includegraphics[width=8cm]{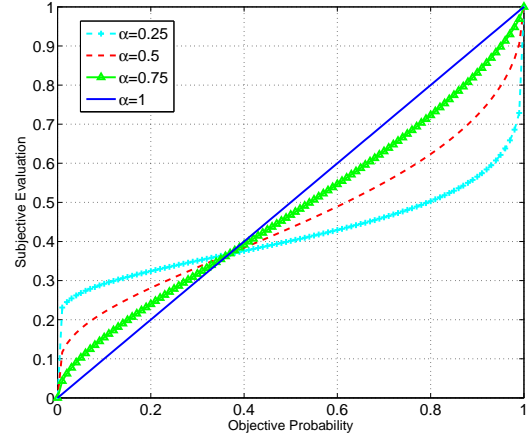}
  \vspace{-0.5cm}
   \caption{\label{fig:probaweight} The subjective probability under the weighting effect, as the objective probability varies.}
\end{center}\vspace{-0.9cm}
\end{figure}

Under PT, the expected utility achieved by a player $i$, given the weighting effect, is \vspace{-0cm}
\begin{equation} \label{eq:multiplayerPT}\vspace{-0cm}
\begin{split}
&U_i^{\text{PT}}( \boldsymbol{p}) = \sum_{\boldsymbol{a} \in \mathcal{A}}\!\! \bigg(p_i(a_i)\!\!\!\!\!\! \prod_{l \in \mathcal{N} \setminus \{i\}}^N\!\!\!\!\!\! w(p_l(a_l))\bigg) u_i(a_i, \boldsymbol{a}_{-i}).
\end{split}
\end{equation}

Here, we assume that a player has a subjective evaluation only of the other players' strategy probabilities. In this respect, customer $i$'s subjective evaluation of its own probability is equal to its objective probability. 

Having defined the cost functions under both EUT and PT, our next goal is to find a solution for the game. Given the set of probability distributions $\mathcal{P}_i$ over its set of strategies $\mathcal{A}_i$, a suitable solution of the game would be a mixed-strategy Nash equilibrium defined as follows:
\begin{definition}
A mixed strategy profile $\boldsymbol{p}^* \in \mathcal{P}=\prod_{i=1}^N \mathcal{P}_i $ is a \emph{mixed strategy Nash equilibrium} (NE) if, for each customer $i \in \{1,2,\ldots, N\}$, we have ($U_i$ represents $U_i^\text{EUT}$ and $U_i^\text{PT}$ under EUT and PT, respectively)\vspace{-0cm}
\begin{equation}
U_i(p_i^*,\boldsymbol{p}_{-i}^*) \le U_i(p_i,\boldsymbol{p}_{-i}^*), \  \forall p_i \in  \mathcal{P}_i.
\end{equation}
\end{definition}
In practice, to avoid slow convergence times and unnecessary overhead, we can consider approximate equilibrium solutions that allow the players to reach the neighborhood of an equilibrium. Hence, we mainly focus on $\epsilon$-Nash equilibria which are defined as follows:
\begin{definition}\label{def:eNE}
A mixed strategy profile $\boldsymbol{p}^* \in \mathcal{P}=\prod_{i=1}^N \mathcal{P}_i $ is an \emph{$\epsilon$-Nash equilibrium ($\epsilon$-NE)} if, for every player $i$, we have\vspace{-0cm}
\begin{equation}\label{eq:epsilonNE}
U_i(p_i^*,\boldsymbol{p}_{-i}^*) \le U_i(p_i,\boldsymbol{p}_{-i}^*)+\epsilon, \  \forall p_i \in  \mathcal{P}_i,
\end{equation}
where $\epsilon$ is a small positive number. 
\end{definition}

In order to find the solution for the game $\Xi$ under both EUT and PT, we find a mixed $\epsilon$-NE in strategic form which represents a point within a close neighborhood of the exact equilibrium~\cite{GT00}.

\vspace{-0cm}

\section{Game Solution and Algorithm}\label{sec:algo}\vspace{-0cm}

In this section, we propose a novel algorithm, under both EUT and PT, to solve the studied DSM game and find an equilibrium point. The proposed algorithm builds on classical fictitious play (FP)~\cite{brown1951iterative}. Thus, we propose the following iterative algorithm to find an $\epsilon$-Nash equilibrium of the proposed game:
\begin{equation} \label{eq:algo}
\boldsymbol{p}_i^{(k+1)}=\boldsymbol{p}_i^{(k)}+\frac{\lambda}{k+1}(\boldsymbol{v}_i^{(k)}-\boldsymbol{p}_i^{(k)}),
\end{equation}
where $0<\lambda<1$ is an inertia weight, $k$ is the iteration index, and $\boldsymbol{v}_i^{(k)}=[v_i^{(k)}(a_1), v_i^{(k)}(a_2), \ldots, v_i^{(k)}(a_H)]^T$ is a vector of player $i$'s strategies such that, $v_i^{(k)}(a_l)=1$ if player $i$ chooses the $l$th strategy at iteration $k$ and $v_i^{(k)}(a_l)=0$ for updating the strategies excluding the $l$th strategy. The pure strategy i.e., the $l$th strategy, is the one that minimizes the expected utility with respect to the updated empirical frequencies. Thus, player $i$ can repeatedly choose its strategy and update $\boldsymbol{v}_i^{(k)}$ as follows:
\begin{equation}\label{eq:FPaction}\vspace{-0.1cm}
v_i^{(k)}(a_l^{(k)})\!=\!
\begin{cases}1, \text{if } a_l^{(k)}\!=\!\arg\min\limits_ {a_i \in \mathcal{A}_i}{u}_{i}(a_i, \boldsymbol{p}_{-i}^{(k-1)}),\!\!\!\!\\
0, \text{otherwise,}
\end{cases}
\end{equation}
where the utility here is the expected value obtained by player $i$ with respect to the mixed strategy of the opponent, when player $i$ chooses pure strategy $a_l$ for EUT and its weighted mixed strategy $w(p_l(a_l))$ for PT. It is well known that our algorithm (as a simplified iterative approach of smooth fictitious play (SFP)) is guaranteed to converge to a mixed $\epsilon$-Nash equilibrium~\cite{fudenberg1995consistency}, as the players' beliefs (probabilistic choices) converge to a fixed point. In particular, a player's belief is its mixed strategy set, and $\epsilon$ represents the utility difference based on beliefs. Here, we define $\epsilon_p$ as the mixed strategy difference corresponding to the utility. SFP can reach an $\epsilon$-NE under EUT, in which the belief difference $\epsilon_p$ between two iterations is small when the number of iterations $k$ goes to infinity. However, to our knowledge, such a result has not been extended to PT. Here, we first prove that the proposed algorithm in (\ref{eq:algo}) will converge to a fixed point in beliefs, as follows:

\begin{theorem}\label{th:cov1}
There exists an inertia weight $\lambda$, $0<\lambda<1$, such that the iterative algorithm in (\ref{eq:algo}) converges to a fixed point in belief, as a mixed $\epsilon_p$-equilibrium ($\epsilon_p \le \frac{\sqrt{2}}{k+1}$).
\end{theorem}
\begin{proof}
For the proposed DSM game, the proposed SFP process is guaranteed to converge to a fixed point in beliefs~\cite{fudenberg1995consistency} under EUT. Here, we mainly prove that, under PT, a player $i$'s iterative sequence $\{\boldsymbol{p}_i(k)\}$ converges to a mixed $\epsilon_p$-equilibrium in beliefs. From (\ref{eq:algo}), we have
%, in which the proof under EUT could be also applied
\begin{equation}\label{eq:pik1}
\begin{split}
\boldsymbol{p}_i^{(k+1)}=&\biggl(1-\frac{\lambda}{k+1}\biggr)\boldsymbol{p}_i^{(k)}+\frac{\lambda}{k+1}\boldsymbol{v}_i^{(k)}\\
=&\biggl(1-\frac{1}{k+1}\biggr)\boldsymbol{p}_i^{(k)}+\frac{1}{k+1}\boldsymbol{v}_i^{(k)}+\frac{1-\lambda}{k+1}\boldsymbol{p}_i^{(k)}\\
&-\frac{1-\lambda}{k+1}\boldsymbol{v}_i^{(k)},
\end{split}
\end{equation}
where the first two terms represent the iteration using FP~\cite{brown1951iterative} and, we define 
\begin{equation}\label{eq:epk}\vspace{-0.1cm}
\epsilon_p=\biggl|\frac{1-\lambda}{k+1}\boldsymbol{p}_i^{(k)}-\frac{1-\lambda}{k+1}\boldsymbol{v}_i^{(k)}\biggr|.
\end{equation}
In (\ref{eq:pik1}) and (\ref{eq:epk}), we present the belief difference between the proposed algorithm and FP, where there exists an $\epsilon_p$ difference at the $k$th iteration. In particular, the distance boundary between FP belief and the approached belief can be bounded as follows:
\begin{equation}\label{eq:epdis}
\begin{split}
\epsilon_p = \frac{1-\lambda}{k+1}| (\boldsymbol{p}_i^{(k)}-\boldsymbol{v}_i^{(k)})| \le  \frac{1-\lambda}{k+1}\cdot \sqrt{2},
\end{split}
\end{equation}
where $\epsilon_p$ becomes small as $k$ increases. Indeed, the Euclidean distance in (\ref{eq:epdis}) involves $\boldsymbol{v}_i^{(k)}$ in (\ref{eq:FPaction}), one of whose components is $1$, and the mixed strategy set $\boldsymbol{p}_i^{(k)}$, whose components will sum to $1$. Thus, the maximum Euclidean distance between $\boldsymbol{p}_i^{(k)}$ and $\boldsymbol{v}_i^{(k)}$ is less than $\frac{\sqrt{2}}{k+1}$ at the $k$th iteration.

Within the given boundary, (\ref{eq:pik1}) can be rewritten as
\begin{equation}\label{eq:pik2}\vspace{-0cm}
\begin{split}
\boldsymbol{p}_i^{(k+1)}\!\!=&\biggl(1-\frac{\lambda}{k+1}\biggr)\boldsymbol{p}_i^{(k)}+\frac{\lambda}{k+1}\boldsymbol{v}_i^{(k)}\\
=&\biggl(1-\frac{\lambda}{k+1}\biggr)\biggl(1-\frac{\lambda}{k}\biggr)\boldsymbol{p}_i^{(k-1)}+\frac{\lambda}{k+1}\boldsymbol{v}_i^{(k)}\\
&+\biggl(1-\frac{\lambda}{k+1}\biggr)\frac{\lambda}{k}\boldsymbol{v}_i^{(k-1)}\\
=&\cdots\\
=&\biggl(1-\frac{\lambda}{k+1}\biggr)\biggl(1-\frac{\lambda}{k}\biggr)\cdots\biggl(1-\frac{\lambda}{2}\biggr)\boldsymbol{p}_i^{(1)}\\
&+\frac{\lambda}{k+1}\boldsymbol{v}_i^{(k)}+\biggl(1-\frac{\lambda}{k+1}\biggr)\frac{\lambda}{k}\boldsymbol{v}_i^{(k-1)}+\cdots\\
&+\biggl(\prod_{j=2}^{k}\bigl(1-\frac{\lambda}{j+1}\bigr)\biggr)\frac{\lambda}{2}\boldsymbol{v}_i^{(1)}.
\end{split}
\end{equation}

The first term in (\ref{eq:pik2}) is $\biggl(\prod_{k=1}^{k=+\infty}(1-\frac{\lambda}{k+1})\biggr)\boldsymbol{p}_i^{(1)}$. Since $\prod_{k=1}^{k=+\infty}(1-\frac{\lambda}{k+1})$ is decreasing as $k$ increases and bounded by $|\boldsymbol{p}_i^{(1)}|\le 1$, the first term is convergent as $k$ goes to infinity. In fact, this term converges to $0$.

For the remaining terms in (\ref{eq:pik2}), we can define
\begin{equation}
\begin{split}
h_j=&\biggl(1-\frac{\lambda}{k+1}\biggr)\biggl(1-\frac{\lambda}{k}\biggr)\cdots\biggl(1-\frac{\lambda}{j+2}\biggr)\frac{\lambda}{j+1}\boldsymbol{v}_i^{(j)}\\
=&\biggl(\prod_{j\le k}\bigl(1-\frac{\lambda}{j+2}\bigr)\biggr)\frac{\lambda}{j+1}\boldsymbol{v}_i^{(j)}.
\end{split}
\end{equation}
Similarly, since $|h_j|$ is decreasing as $k$ increases and bounded by $|\boldsymbol{v}_i^{(j)}|=1$, $\sum_{j=1}^{k} h_j$ is convergent as $k$ goes to infinity. Thus, we conclude that, $\boldsymbol{p}_i^{(k+1)}$ will converge to a fixed point, as a mixed $\epsilon_p$-equilibrium in beliefs using (\ref{eq:algo}).
\end{proof}
\begin{remark} \label{rmk:fp}
When $\lambda=1$, the proposed algorithm in (\ref{eq:pik1}) is reduced to FP. Hence (\ref{eq:pik1}) and (\ref{eq:pik2}) can be derived as $\boldsymbol{p}_i^{(k+1)}
=\frac{1}{k+1}\boldsymbol{p}_i^{(1)}+\frac{1}{k+1}(\boldsymbol{v}_i^{(k)}+\cdots+\boldsymbol{v}_i^{(1)})$. Thus, FP ($\lambda=1$) might cycle if $\boldsymbol{v}_i^{(k)}$ repeats after some iteration, i.e., $\boldsymbol{v}_i^{(2)}=\boldsymbol{v}_i^{(4)}=\boldsymbol{v}_i^{(6)}=\cdots$, and $\boldsymbol{v}_i^{(1)}=\boldsymbol{v}_i^{(3)}=\boldsymbol{v}_i^{(5)}=\cdots$, while the proposed algorithm ($0<\lambda<1$) in (\ref{eq:algo}) converges to a fixed point.
\end{remark}

\begin{theorem}\label{th:cov2}
For the proposed DSM game, the proposed algorithm in (\ref{eq:algo}) is guaranteed to converge to a mixed $\epsilon$-NE under both EUT and PT, as its beliefs converge to a mixed $\epsilon_p$-equilibrium.
\end{theorem}

The proof of Theorem~\ref{th:cov2} is given in the appendix.

\begin{table}[!t]\vspace{-0.2cm}
  \centering
  \caption{
    \vspace*{-0.1em}Proposed Load Shifting Solution}\vspace*{-0.5em}
    \begin{tabular}{p{8cm}}
      \hline\vspace*{+0.05em}
      
\textbf{Phase 1 - Proposed Dynamics:}   \vspace*{.1em}\\
\hspace*{1em}Each customer $i \in \mathcal{N}$ chooses a starting strategy $\boldsymbol{p}_i^{\textrm{init}}$ as its mixed strategy of participation.\vspace*{.1em}\\
\hspace*{2em}\textbf{repeat,}\vspace*{.2em}\\
\hspace*{1em}1) Each customer $i \in \mathcal{N}$ observes others' participation $\boldsymbol{p}_{-i}$\vspace*{.1em}\\
\hspace*{1em}2) Each customer $i\! \in\! \mathcal{N}$ updates its better response strategy using\\
\hspace*{1em}the proposed algorithm in (\ref{eq:algo}):\\
\hspace*{3em}i) The utility operator communicates with the customer\vspace*{.1em}\\
\hspace*{3em}using the grid's two-way architecture communication\vspace*{.1em}\\
\hspace*{3em}(see \cite{EW03} or \cite{hossain2012smart} and references therein).\vspace*{.1em}\\
\hspace*{3em}ii) Customers' loads can be shifted via DSM as in Section~\ref{sec:ls}.\vspace*{.1em}\\
\hspace*{4em}\textbf{Load Shedding}\vspace*{.1em}\\
\hspace*{4em}a) The utility operator advertises all customers' participation\vspace*{.1em}\\
\hspace*{4em}times $a_i \in \mathcal{A}_i$ using their mixed strategies.\vspace*{.1em}\\
\hspace*{4em}b) Each customer publishes its participation $p_i^\text{PT}$ based on a \\
\hspace*{4em}mixed strategy $\boldsymbol{p}^\text{EUT}$, representing the subjective observation\\
\hspace*{4em}with underlying weight effect.\vspace*{.1em}\\
\hspace*{4em}c) After combining probabilities as in  (\ref{eq:multiplayerET}), customer $i$ will\\
\hspace*{4em}receive an expected cost under EUT, sent by the operator,\\
\hspace*{4em}and observe the current vector of strategies $\boldsymbol{p}_{-i}$ so as to\\
\hspace*{4em}assess its subjective utility in (\ref{eq:multiplayerPT}) under PT.\vspace*{.1em}\\
\hspace*{4em}d) Customer $i$ chooses its strategic response in (\ref{eq:algo}). \vspace*{.1em}\\
\hspace*{2em}\textbf{until} convergence to a mixed NE strategy vector $\boldsymbol{p}^*$.\vspace*{.2em}\\
\textbf{Phase 2 - Power Company Strategy}   \vspace*{.1em}\\
\hspace*{1em}1) The operator receives the participation information\vspace*{.1em}\\
\hspace*{1em}given the mixed strategy set as per $\boldsymbol{a}^*$.\vspace*{.1em}\\
\hspace*{1em}2) Actual load shifting is implemented under realistic participation.\vspace*{.1em}\\
   \hline
    \end{tabular}\label{tab:algo}\vspace{-0.7cm}
\end{table}

We summarize the proposed DSM solution in Table~\ref{tab:algo}. To find the solution of this game under both EUT and PT, we use the algorithm in (\ref{eq:algo}) to solve for the mixed $\epsilon$-NE. In the first phase of the algorithm, each customer will set an initial probability vector, while any non-participating customers will set its probability to $0$ throughout the whole DSM process. At the beginning, each participating customer observes others' strategies and evaluates its expected costs of every strategy. In the evaluation process, a participating customer can receive an estimate of the expected costs from the grid operator under EUT, since the grid operator objectively collects information and provides required services/information. Then, the customers will overweight or underweight the information about others' strategies based on the individual weighting effect in (\ref{eq:weight}). Moreover, each customer will subjectively estimate its expected costs of every strategy and then frequently report its participation based on its mixed strategy. The communication between customers and grid operator will continue until the expected costs satisfy the setting $\epsilon$ in utility (\ref{eq:multiplayerET}), (\ref{eq:multiplayerPT}) and (\ref{eq:epsilonNE}). Once a mixed $\epsilon$-NE is reached, the customers will frequently signal their participation decisions based on probabilities and participate in demand-side management in Phase $2$. This phase of the proposed load shifting solution is the practical market operation. Given the submitted information, the grid operator will shift/reduce local loads. The actual process of Phase $2$ is beyond the scope of this paper and will follow real-world contract negotiations that could be interesting to study in future work.

\vspace{-0.1cm}
\section{Simulation Results and Analysis}\label{sec:sim}\vspace{-0cm}
For our simulations, we use the real-world load profile in~\cite{web} which represents customers' initial demands, i.e., the data between April $29$th, 2013, and May $4$th, 2013, from Miami International Airport, since the data of local customer houses, or groups, is large, private and confidential. In all simulations, each customer can choose a starting time to participate in DSM from the time period between $18$:$00$ and $20$:$00$. Alternatively, the customer can decide not to participate; that is, $\mathcal{A}_i = \{18, 19, 20, 24\}, \forall i \in \mathcal{N}$. Also, we set $\beta=0.86$ and $\gamma=0.6$ for all customers.

\begin{figure}[!t]
 \begin{center}
 \vspace{-0.2cm}
  \includegraphics[width=8cm]{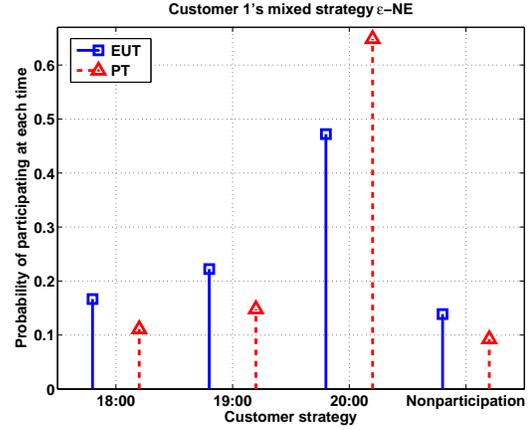}
 \vspace{-0.5cm}
   \caption{\label{fig:histP1} The probability performance of all mixed strategic participations for Customer $1$ under both EUT and PT.}
\end{center}\vspace{-0.2cm}
\end{figure}

\begin{figure}[!t]
 \begin{center}
 \vspace{-0.2cm}
  \includegraphics[width=8cm]{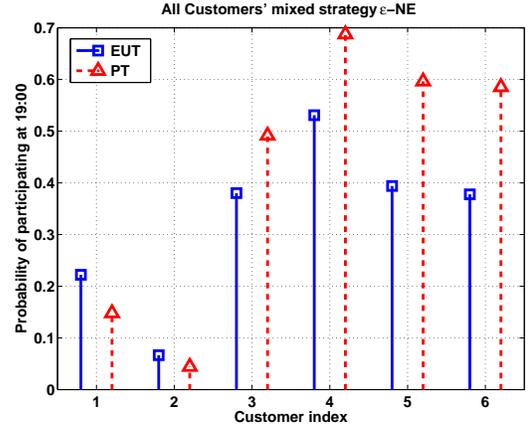}
 \vspace{-0.5cm}
   \caption{\label{fig:hist1900} The mixed strategic participations for the six customers using both EUT and PT at $19$:$00$.}
\end{center}\vspace{-0.9cm}
\end{figure}

Fig.~\ref{fig:histP1} shows, for a smart grid with six customers, all mixed strategies of a selected Customer $1$ under both EUT and PT. In this figure, we choose $\alpha=0.7$ for all customers to represent their distortion under the weighting effect in (\ref{eq:weight}). Clearly, Customer 1's PT strategy is different than its EUT strategy. Compared to the EUT results (solid lines), the PT participating strategy (dash line) at $20$:$00$ is larger, while the other PT mixed strategies are smaller. In essence, we observe that PT generally enhances (reduces) high (low) frequency EUT strategies. For instance, under EUT, Customer $1$ will have the highest participation strategy at $20$:$00$ due to load shifting as per (\ref{eq:redde}) and the varying price in (\ref{eq:price}). In particular, its largest mixed strategy is greater than the average mixed strategy, i.e., $0.25$ in the proposed four-strategy game. Second, because of the weighting effect in (13), in Fig.~\ref{fig:histP1}, we can see how PT behavior can be different from EUT. In particular, we observe that a PT customer wants to participate more than in the EUT case at times when it has a low hourly energy prices in (\ref{eq:price}), which also corresponds to a situation with high participation under EUT. In other words, for minimizing its cost, under PT, each customer tends to increase its participation at the hours with lower hourly prices in (\ref{eq:price}) and to decrease its participation at the hours with higher hourly prices. This is mainly due to the relationship between shifting the demand and the resulting price as captured by (\ref{eq:price}) and (\ref{eq:utility}). Thus, under realistic behavioral considerations, customers will accentuate the traditional behavior perceived under EUT. This will consequently impact the overall DSM performance as shown in the next figures. Accordingly, under PT, we can observe that higher EUT probabilities will become more pronounced. Thus, the largest mixed strategy using EUT, i.e., at $20$:$00$, would be overweighted via PT observation, and vice versa. This is due to the fact that each PT customer takes a risk by participating when the hourly costs are low.

Fig.~\ref{fig:hist1900} shows, using the same parameters as Fig.~\ref{fig:histP1}, the probability that each customer participates in DSM at $19$:$00$ under both EUT and PT. In this figure, we can first see that the mixed strategy of each customer using PT is different from that resulting from EUT. Under EUT, the rational mixed strategies made by some customers, such as Customers $3$-$6$, are higher than the average mixed strategy, because they have low costs at $19$:$00$. In particular, given the price in (\ref{eq:price}), a customer's lowest demand between $18$:$00$ and $20$:$00$ would cause the lowest costs and the highest participation. This implies that a rational customer wants to participate in DSM when it does not require a lot of energy (or its demand is low) in practice, since the payment is low. Under PT, the customers' realistic decisions will impact the price in (\ref{eq:price}) and change their participation levels. In Fig.~\ref{fig:hist1900}, some customers' mixed strategic components using PT are greater than those of EUT, i.e., Customers $3$-$6$, while Customers $1$ and $2$ have a lower participation probability using PT. If four PT customers, based on their loads between $18$:$00$ and $24$:$00$, decide to shift more load at $19$:$00$, their PT demands will be less than EUT. Then, compared to EUT, the PT prices of Customers $1$ and $2$ will increase because of the changing load fraction in (\ref{eq:price}). Because of the increasing prices, Customers $1$ and $2$'s costs increase and they will be less interested in participating at $19$:$00$. Indeed, there will exist some customers, such as Customers $1$ and $2$, who want to decrease their participation at $19$:$00$ under PT as shown in Fig.~\ref{fig:histP1}. Then, the overall demand increases and Customers $3$-$6$ obtain a decreasing pricing in (\ref{eq:price}) which will make these four customers more likely to participate at 19:00. 

\begin{figure}[!t]
 \begin{center}
 \vspace{-0.2cm}
  \includegraphics[width=8cm]{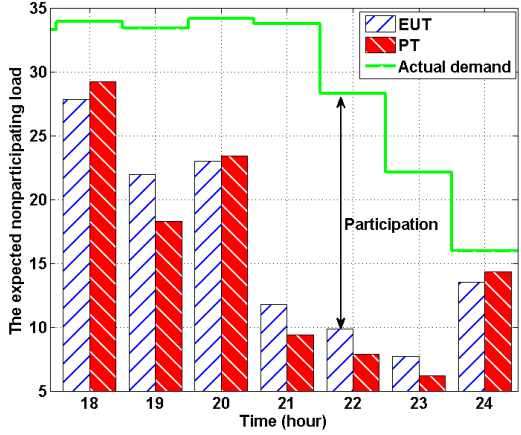}
 \vspace{-0.5cm}
   \caption{\label{fig:nonload} The expected nonparticipating load for the six customer game under mixed strategies using both EUT and PT over $24$ hours, when all customers have the same $\alpha=0.7$.}
\end{center}\vspace{-0.9cm}
\end{figure}

Fig.~\ref{fig:nonload} shows the expected nonparticipating load profile in the proposed DSM game as time varies. We choose the same customers as in Fig.~\ref{fig:histP1} with the distortion parameters are set to $\alpha = 0.7$ for all customers. The actual demand is the summation of initial demands in (\ref{eq:actuald}). In Fig.~\ref{fig:nonload}, the nonparticipating load is the minimum expected load that the power company must supply, while the participating load represents the load that can be partly shifted based on each individual customer's action in (\ref{eq:redde}). Here, all customers can start the game from $18$:$00$ and the expected nonparticipating loads between EUT and PT are different. On the one hand, the difference between EUT and PT during $18$:$00$ and $20$:$00$ is due to the change in the customers' decisions, as previously shown in Fig.~\ref{fig:hist1900}. On the other hand, between $21$:$00$ and $23$:$00$, the customers nonparticipating load using PT is always less than that using EUT. Indeed, if we translate the proposed game into a ``participate or not participate'' game, $21$:$00$ can be used to distinguish two such strategies, and customer behaviors before $21$:$00$ can impact their participation at a later time. However, the fact that the nonparticipation level for PT is less than that for EUT after $21$:$00$ directly relates to the choice of a distortion parameter $\alpha$. In particular, a small deviation from the rational strategy for $\alpha=0.7$ leads to an increased competition between the customers due to the fact that a weighted observation increases the costs in (\ref{eq:utility}). Such small deviation represents a case in which a customer participates in DSM but does not trust its view of the opponents' strategy or the information received from the power company. Thus, a slight deviation from the rational path causes increasing costs and customers are more apt to shift their loads under PT to decrease the impact of being non-rational, compared to EUT. As a result, after $21$:$00$, as seen in Fig.~\ref{fig:nonload}, the PT nonparticipating load will be less than that for EUT. At $24$:$00$, the power company needs to deal with the remaining load as defined in (\ref{eq:later2}).

\begin{figure}[!t]
 \begin{center}
 \vspace{-0.2cm}
  \includegraphics[width=8cm]{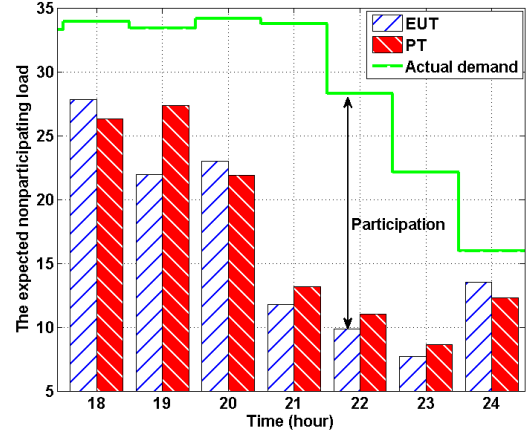}
 \vspace{-0.5cm}
   \caption{\label{fig:nonloaddifal} The expected nonparticipating load for the six customer game under mixed strategies using both EUT and PT over $24$ hours, when customers have different values of $\alpha$.}
\end{center}\vspace{-0.9cm}
\end{figure}

Compared to Fig.~\ref{fig:nonload}, Fig.~\ref{fig:nonloaddifal} shows the expected nonparticipating load profile using different distortion parameters, as time varies. The distortion parameter $\alpha$ in (\ref{eq:weight}) allow us to measure how a customer perceives the actions of its opponents. A large $\alpha$ implies a small distortion, while a small $\alpha$ represents an excessively subjective perception. In particular, we choose $\alpha=[0.5\ 0.5\ 0.2\ 0.1\ 0.1\ 0.1]^T$. In this figure, we can see that, when some customers have a very irrational observation of their opponents, the PT nonparticipating load between $21$:$00$ and $23$:$00$ will be higher than EUT. This implies that, in reality, if some customers deviate significantly from their rational strategies (for example, a customer forgets to assist the power company in load shifting), the power company will not be able to shift the total load predicted by the rational, objective model. Thus, the power company can use the distortion parameter to decide on how to design its DSM scheme.

\begin{figure}[!t]
 \begin{center}
 \vspace{-0.2cm}
  \includegraphics[width=8cm]{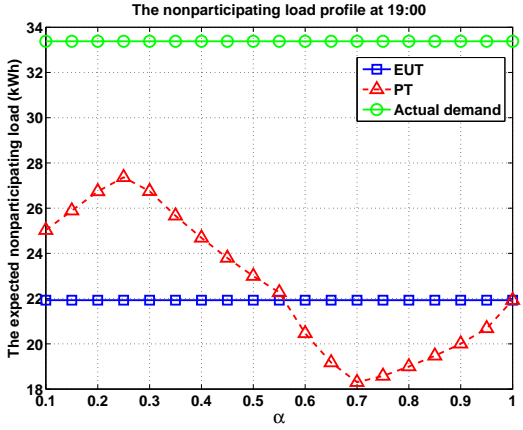}
 \vspace{-0.5cm}
   \caption{\label{fig:alpha1900} The expected nonparticipating load of all customers at $19$:$00$ as $\alpha$ varies.}
\end{center}\vspace{-0.9cm}
\end{figure}

Fig.~\ref{fig:alpha1900} shows the expected nonparticipating load at $19$:$00$, as the weighting effect parameter $\alpha$ varies. In this figure, we can see that the expected nonparticipating load under EUT is $65.7\%$ of the total load and, the nonparticipating load under PT is less than EUT when $\alpha>0.56$. This is because the majority of PT customers have more interest in participating at $19$:$00$, as shown in Fig.~\ref{fig:hist1900}. Thus, the power company can shift more load in practice, compared to EUT. Also, this figure shows that there exists a distortion threshold, such that, if $\alpha$ is greater (smaller) than the threshold, PT customers will have lower (higher) nonparticipating loads than EUT cases. A large distortion parameter, or a small deviation from EUT, yields an increased competition thus raising the costs to the customers. Consequently, the customers will become risk seeking and more apt to shift their loads and decrease their costs. Thus, the increasing PT costs will force the majority to shift more loads, compared to EUT. However, a small distortion parameter, or a large deviation from EUT, will lead to highly distorting behavior from the customers which will lead to increasingly high competition and decreasing participation, as customers become extremely risk averse and unwilling to participate in the DSM process. In a nutshell, Fig.~\ref{fig:alpha1900} shows that a small deviation from EUT may be beneficial for the power company as it increases customers' participation. In contrast, a significant deviation from EUT will inevitably lead to highly risk averse behavior which will prevent most customers from participating; thus yielding detrimental results for the grid and preventing the operator from reaping the benefits of DSM.

\begin{figure}[!t]
 \begin{center}
 \vspace{-0.2cm}
  \includegraphics[width=8cm]{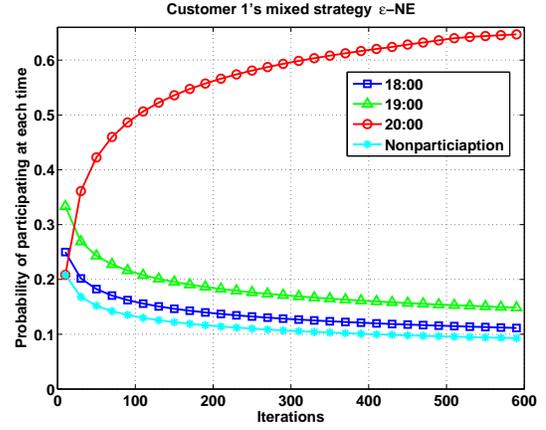}
 \vspace{-0.5cm}
   \caption{\label{fig:iter} The probability performance of mixed strategy for six customers based on EUT and PT.}
\end{center}\vspace{-0.2cm}
\end{figure}

In Fig.~\ref{fig:iter}, we show all PT mixed strategies of Customer $1$ (corresponding to Fig.~\ref{fig:histP1}), as the number of iterations increases. The proposed algorithm can be shown to converge \emph{sublinearly} by studying its rate of convergence which is defined as $R=\lim\limits_{k \rightarrow \infty}\frac{|\boldsymbol{p}^{(k+1)}-\boldsymbol{p}^*|}{|\boldsymbol{p}^{(k)}-\boldsymbol{p}^*|}$. Here, we mainly test the convergence of the proposed algorithm in (\ref{eq:algo}). In particular, Customer $1$'s initial probability is $\boldsymbol{p}_1^{(0)}=[0.2500\ 0.3333\ 0.2083\ 0.2083]^T$ and it converges to $\boldsymbol{p}_1^{*}=[0.1110\ 0.1480\  0.6484\ 0.0925]^T$. From Fig.~\ref{fig:histP1}, due to the low hourly costs, Customer $1$ will increase its participation at $20$:$00$, corresponding to the highest EUT frequency. In Fig.~\ref{fig:iter}, since the initial probability at $20$:$00$ is small, i.e., $p_1^{(0)}(a_3)=0.2083$, Customer $1$ will choose its best strategy and update $v_1(a_3)=1$ in (\ref{eq:FPaction}). Thus, we can see that the participating probability at $20$:$00$ increases, as the iteration $k$ increases. Similarly, because of the large initial probabilities at other strategies, Customer $1$ will decrease underlying participation via updating $v_1(a_1)=v_1(a_2)=v_1(a_4)=0$ in (\ref{eq:FPaction}). Also, Table~\ref{tab:multi} shows Customer $1$'s all mixed $\epsilon$-Nash equilibria under both EUT and PT. Using different initial probability vectors $\boldsymbol{p}^{(0)}$, the proposed algorithm in (\ref{eq:algo}) will reach different Nash equilibria, and the differences between EUT and PT at two different Nash equilibria are not the same.

\begin{table}[!t]\vspace{-0cm}
\scriptsize
  \centering
  \caption{
    \vspace*{-0em} Customer $1$'s All Mixed $\epsilon$-NE under both EUT and PT}\vspace*{-0.3cm}
\begin{tabular}{|c|c|c|}
\hline
 & EUT & PT  \\ [0.5ex]\hline
$1$ & $[0.1667\ 0.2222\ 0.4722\ 0.1389]^T$ & $[0.1110\ 0.1480\ 0.6484\ 0.0925]^T$ \\\hline
$2$ & $[0.2667\ 0.1333\ 0.5333\ 0.0667]^T$ & $[0.1776\ 0.0888\ 0.6891\ 0.0444]^T$ \\\hline
$3$ & $[0.0392\ 0.1569\ 0.4509\ 0.3530]^T$ & $[0.0261\ 0.1045\ 0.6343\ 0.2351]^T$ \\\hline
$4$ & $[0.1159\ 0.2609\ 0.5072\ 0.1159]^T$ & $[0.0772\ 0.1738\ 0.6718\ 0.0772]^T$\\\hline
$5$ & $[0.2222\ 0.0889\ 0.4666\ 0.2222]^T$ & $[0.1480\ 0.0592\ 0.6447\ 0.1480]^T$\\\hline
\end{tabular}
\label{tab:multi}\vspace{-0.7cm}
\end{table}

\vspace{-0.1cm}

\section{Conclusions}\label{sec:conc}\vspace{-0cm}
In this paper, we have introduced a novel approach for
studying the problem of demand-side management. It is of interest to factor in explicitly the behavior of customers in DSM. In particular, we have developed a game-theoretic approach, based on prospect theory, using which each player subjectively observes other players' actions and determines its own actions so as to minimize a cost function that captures the electricity cost over $24$ hours. Then, we have proposed an algorithm and have proved that it reaches a mixed $\epsilon$-NE. Simulation results have shown that deviations from classical, objective game-theoretic DSM mechanisms can lead to unexpected results and loads on the grid, depending on the level of the customers' subjective perceptions of each others' actions. Therefore, whether a customer participates in DSM or not, depends on this level of rationality of the customer. In a nutshell, the results of this paper have provided important insights into the factors underlying the modest participation in DSM schemes observed in real-world smart grid systems.

\vspace{-0.1cm}

\appendix
\vspace{-0.1cm}

\subsection{Proof of Theorem~\ref{th:cov2}}

\vspace{-0.1cm}
\begin{proof}
The convergence of the beliefs to a mixed $\epsilon_p$-equilibrium is shown in Theorem~\ref{th:cov1}. Based on Theorem~\ref{th:cov1}, the convergence of a mixed $\epsilon$-NE under EUT is a known result for SFP in~\cite{fudenberg1995consistency}. The following proof mainly focuses on PT and we prove this theorem by contradiction.

Suppose that $\{\boldsymbol{p}_k\}$ is an iterative process resulting from  the proposed algorithm that converges to a mixed strategy $\boldsymbol{p}^*$ (an $\epsilon_p$ equilibrium) after $k$ iterations. Also, assume a mixed NE $\boldsymbol{\delta}^*=\{\boldsymbol{\delta}^*_i,\boldsymbol{\delta}^*_{-i}\}$ is near to the fixed point $\boldsymbol{p}^*$. Based on contradiction, if the point $\boldsymbol{p}^*=\{\boldsymbol{p}^*_i,\boldsymbol{p}^*_{-i}\}$ is not an $\epsilon$-mixed NE, there must exists a strategy $p'_i(a'_i) \in \boldsymbol{p}^*_i$, such that {1)} $p_i(a_i)>0, p_i(a_i) \in \boldsymbol{p}^*$ (at least one mixed strategy of player $i$ is not zero), and {2)} \vspace{-0.1cm}
\begin{equation*}\vspace{-0.2cm}
\begin{split}
u_i^{\text{PT}}\biggl(a_i,\boldsymbol{p}^*_{-i}\biggr)>&u_i^{\text{PT}}\biggl(a'_i,\boldsymbol{p}^*_{-i}\biggr),\\
u_i^{\text{PT}}\biggl(a_i,\boldsymbol{\delta}^*_{-i}\biggr)+\epsilon_i\biggl(a_i,\boldsymbol{\delta}^*_{-i}\biggr)>&u_i^{\text{PT}}\biggl(a'_i,\boldsymbol{\delta}^*_{-i}\biggr)+\epsilon_i\biggl(a'_i,\boldsymbol{\delta}^*_{-i}\biggr),\\
\end{split}
\end{equation*}
where $u_i(a_i,\boldsymbol{p}^*_{-i})$ is the expected utility of pure strategy $a_i$ and $\epsilon_i(a_i,\boldsymbol{\delta}^*_{-i})$ is the utility difference between NE and $\epsilon$-NE. Here, since the iterative mixed strategy decreases as the number of iterations $n$ increases ($n \le k$), the utility distance $\epsilon_s$ of a pure strategy between two neighboring iterations/steps must be less than a value after a certain iteration $k$. In particular, if the proposed algorithm converge to a fixed $\epsilon_p$ equilibrium at iteration $k$ and it is not an $\epsilon$-NE in belief (utility), we can find an $\epsilon_s$, such that\vspace{-0.1cm}
\begin{equation*}\vspace{-0.1cm}
\begin{split}
0<\epsilon_s<&\frac{1}{2}\biggl|u_i^{\text{PT}}(a_i,\boldsymbol{p}^*_{-i})-u_i^{\text{PT}}(a'_i,\boldsymbol{p}^*_{-i})\biggr|,\\
0<2\epsilon_s<&\biggl(u_i^{\text{PT}}(a_i,\boldsymbol{\delta}^*_{-i})+\epsilon_i(a_i,\boldsymbol{\delta}^*_{-i})\biggr)\\
&-\biggl(u_i^{\text{PT}}(a'_i,\boldsymbol{\delta}^*_{-i})+\epsilon_i(a'_i,\boldsymbol{\delta}^*_{-i})\biggr),
\end{split}
\end{equation*}
where $u_i^{\text{PT}}(a_i,\boldsymbol{\delta}^*_{-i})=\sum_{\boldsymbol{a} \in \mathcal{A}} u_i(a_i,\boldsymbol{a}_{-i}^*) w_i(\boldsymbol{\delta}^*_{-i})$. For $n \ge k$, the proposed algorithm process satisfies\vspace{-0.2cm}

\begin{equation*}\label{eq:ineq}\vspace{-0.1cm}
\begin{split}
u_i^{\text{PT}}(a_i,\boldsymbol{p}_{-i}^n)=&\sum_{\boldsymbol{a} \in \mathcal{A}} u_i(a_i,\boldsymbol{a}_{-i}^n) w_i(\boldsymbol{p}_{-i}^n)\\
\ge &\biggl(\sum_{\boldsymbol{a} \in \mathcal{A}} u_i(a_i,\boldsymbol{a}_{-i}^*) w_i(\boldsymbol{p}^*_{-i})\biggr)- \epsilon_s\\
=&\biggl(\sum_{\boldsymbol{a} \in \mathcal{A}} u_i(a_i,\boldsymbol{a}_{-i}^*) w_i(\boldsymbol{\delta}^*_{-i})\biggr)+\epsilon_i(a_i,\boldsymbol{\delta}_{-i}^*)- \epsilon_s\\
>&\biggl(\sum_{\boldsymbol{a} \in \mathcal{A}} u_i(a'_i,\boldsymbol{a}_{-i}^*) w_i(\boldsymbol{\delta}^*_{-i})\biggr)+\epsilon_i(a'_i,\boldsymbol{\delta}_{-i}^*)+ \epsilon_s\\
=&\biggl(\sum_{\boldsymbol{a} \in \mathcal{A}} u_i(a'_i,\boldsymbol{a}_{-i}^*) w_i(\boldsymbol{\delta}^*_{-i})\biggr)+ \epsilon_s\\
\ge &\sum_{\boldsymbol{a} \in \mathcal{A}} u_i(a'_i,\boldsymbol{a}_{-i}^n) w_i(\boldsymbol{p}^n_{-i})\\
=&\ u_i^{\text{PT}}(a'_i,\boldsymbol{p}_{-i}^n).
\end{split}
\end{equation*}
Thus, player $i$ would not play $a_i$ but $a'_i$ after the $n$th iteration, and we will have $p_i(a_i)=0$ and $\boldsymbol{w}_{-i}(p_i(a_i))=0$ (the other player's observation on $p_i$). Here, we have proved that, $p_i(a_i)=0$ contradicts the initial assumption $p_i(a_i)>0$; thus the theorem follows.
\end{proof}

\def\baselinestretch{0.8}
\bibliographystyle{IEEEtran}
\bibliography{references}

% that's all folks
\end{document}